\documentclass[12pt,titlepage]{utarticle}

\usepackage{amsmath}
\usepackage{amsfonts}
\usepackage{amssymb}
\usepackage{amsthm}
\usepackage{graphicx}
\usepackage{ucs}
\usepackage{hyperref}

\newcommand{\lt}{<}

\newtheorem*{prop}{Proposition}

\begin{document}
\preprint{MIFP--07--19\\ UTTG-03-07\\}

\title{Wormholes in Maximal Supergravity}

\author{Aaron Bergman
        \thanks{Research supported by NSF Grants PHY-0505757, PHY-0555575 and by Texas A\&M Univ.}
        \address{George P. \& Cynthia W. Mitchell Institute for
         Fundamental Physics\\
        Texas A\&M University\\
        College Station, TX 77843-4242\\
        \email{abergman@physics.tamu.edu}}
        \\
        Jacques Distler
        \thanks{Research supported in part by NSF Grant PHY-0455649.}
        \address{Theory Group, Physics Department\\
        University of Texas at Austin\\
        Austin, TX 78712\\
        \email{distler@golem.ph.utexas.edu}}}

\Abstract{In this brief note, we reconsider the problem of finding Euclidean wormhole solutions to maximal supergravity in $d$ dimensions. We find that such solutions exists for all $d\leq 9$. However, we argue that, in toroidally-compactified string theories, these saddle points never contribute to the path integral because of a tension with U-duality.}

\maketitle
\newpage


\hypertarget{intro}{}\section{{Introduction}}\label{intro}

The relation of the path integral formulation of quantum gravity to string theory has long been mysterious. While both the Euclidean and Lorentzian path integrals for gravitational theories have longstanding pathologies, by analogy with well-understood situations in quantum field theory, they can be used to inspire effects that may or may not be present in a true theory of quantum gravity. Thus, it is an extremely interesting question to see if such effects are present in string theory.

Towards this end, Arkani-Hamed \emph{et al} in \cite{AOP} examined the situation of Euclidean wormhole \cite{GS1,LRT,H1,H2}. These should represent stationary points of the Euclidean path integral of quantum gravity. One can try to understand them either through Wick rotation or as contributions to the path integral obtained by ``deforming the contour'' analogously to the usual stationary phase approximation for finite dimensional integrals. One might have naively thought such solutions to have led to bilocal terms in the effective action, but Coleman \cite{Co} (and also \cite{GS2}) instead reinterpreted them in terms of modifications of the coupling constants of local terms in the lagrangian. Arkani-Hamed \emph{et al} demonstrate the existence of such wormhole solutions in compactifications of string theory on higher dimensional tori and argue that the contribution of such solutions represents a contradiction with the predictions of the AdS/CFT conjecture.

The purpose of this note is two-fold. First, we demonstrate further wormhole solutions that were missed in \cite{AOP} and in earlier works \cite{MM,GS3,Ta}. In particular, we are able to find wormhole solutions in toroidal compactifications of Type II on $T^{10-d}$ for all $d\leq 9$ and -- more generally -- in any compactification preserving enough supersymmetry such that the scalars in the gravity multiplet take their values in a Riemannian symmetric space.

String theory, in contrast to supergravity, possesses a discrete gauge symmetry, termed U-duality, which reduces the true moduli space of the scalars to a quotient of the symmetric space (termed a locally symmetric space). Our second goal is to argue that none of the wormhole solutions (those presented in \cite{AOP} or the new ones presented here) can be assigned well-defined transformation properties under U-duality. It, therefore, seems unlikely that there exists a procedure (involving, say, summing over wormhole configurations) which would be compatible with U-duality invariance. Hence, we conclude that these wormholes cannot contribute to the quantum gravity path integral in any dimension, generalizing the result of \cite{AOP}.

\hypertarget{constructing_the_solution}{}\section{{Constructing the Solution}}\label{constructing_the_solution}

\hypertarget{background}{}\subsection{{Background}}\label{background}

We will consider Type II string theory, compactified on $T^{10-d}$. The low energy physics is governed by maximal supergravity in $d$ dimensions. At the level of the supergravity, there is a continuous symmetry group, $G$, which is the split real form of some semisimple Lie group. The scalar fields take values in a nonlinear $\sigma$-model whose target space is the symmetric space, $\mathcal{M} = G/K$, where $K$ is the maximal compact subgroup of $G$. $G$ acts as the group of continuous isometries of this space. The various $G$ and $K$, for different choices of $d$ are listed in the table below.

\begin{center}
\begin{tabular}{|r|c|c|c|c|c|c|c|}
\hline
d&3&4&5&6&7&8&9\\
\hline\hline
G&$E_{8,8}$&$E_{7,7}$&$E_{6,6}$&$Spin(5,5)$&$SL_{5}(\mathbb{R})$&$SL_{3}(\mathbb{R})\times SL_{2}(\mathbb{R})$&$SL_{2}(\mathbb{R})$\\
\hline
K&$Spin(16)$&$SU(8)$&$Sp(4)$&$Spin(5)\times Spin(5)$&$Spin(5)$&$SU(2)\times SO(2)$&$SO(2)$\\
\hline
\end{tabular}
\end{center}

In string theory, the continuous $G$-symmetry is broken by higher-derivative corrections to the low-energy supergravity. What remains is a discrete group, $G(\mathbb{Z})$, which acts as a discrete gauge symmetry. Thus, the the correct moduli space is the quotient $\mathcal{M}_{\text{true}} = G(\mathbb{Z})\,\backslash\: G /K$. For the purpose of finding solutions to the supergravity, we can ignore these discrete identifications and work on the covering space $\mathcal{M}=G/K$. We will discuss the implications of U-duality in the following section.

We begin by reviewing the work of \cite{AOP}. There, Arkani-Hamed \emph{et al} look for wormhole solutions of the supergravity theory. In particular, they take as an ansatz an $O(d)$-invariant metric of the form

\begin{equation}
d r^2 + a^2(r) d\Omega_{d-1}^2\ .
\label{metricAnsatz}\end{equation}
The equations of motion for the scalars coupled to gravity are

\begin{equation}
\begin{split}
\frac{{a'}^2}{a^2}- \frac{1}{a^2} \frac{G_{i j}\phi^{i\prime}\phi^{j\prime}}{2(d-1)(d-2)} &=0\ ,\\
\left(a^{d-1} G_{i j} \phi^{j\prime}\right)' - \frac{1}{2} a^{d-1} G_{j k,i} \phi^{j\prime} \phi^{k\prime} &= 0
\end{split}
\label{eom}\end{equation}
where $G_{i j}$ is the metric on the scalar manifold, $\mathcal{M}$, and $'$ indicates derivative with respect to $r$.

If we define $\tau$ via

\begin{equation}
\frac{d r}{d\tau} = a(r)^{d-1}
\label{tauDef}\end{equation}
and use \eqref{eom}, we recognize that the second equation is proportional to the geodesic equation on $\mathcal{M}$. Thus, the scalars travel along a geodesic, and a constant of the motion is given by
\begin{equation}
C = a^{2(d-1)} G_{i j}\phi^{i\prime}\phi^{j\prime}\ .
\label{cdef}\end{equation}
A wormhole solution (in flat space) has the following asymptotics:
\begin{displaymath}
\begin{gathered}
   a(r)^2 \sim r^2,\quad r\to \pm\infty\ ,\\
   a(r=0) = a_0\ .
\end{gathered}
\end{displaymath}
These constraints are satisfied if we set $C$ in \eqref{cdef} to

\begin{equation}
C = -2(d-1)(d-2) a_0^{2(d-2)} \lt 0\ .
\label{cVal}\end{equation}
There is a further constraint, however. In \cite{AOP}, this ansatz is used to calculate the distance the moduli must travel between the two asymptotic regions of the wormhole solution:

\begin{displaymath}
\begin{aligned}
D[\phi(r=+\infty), \phi(r=-\infty)]&= 2 D[\phi(r=+\infty), \phi(r=0)]\\ 
&= \pi \sqrt{\frac{2(d-1)}{d-2}}\ .
\end{aligned}
\end{displaymath}
Thus, in order to admit a wormhole solution, $\mathcal{M}$ must possess a timelike geodesic of at least this length,

\begin{equation}
\Delta\tau \geq \pi \sqrt{\frac{2(d-1)}{d-2}}\ .
\label{bound}\end{equation}
Of course, since the scalar manifold, $\mathcal{M}= G/K$, in the supergravity theory is Riemannian, everything is spacelike, and there are no \emph{real} wormhole solutions. Instead, we look for complex saddle points. That is, we consider Wick-rotating one (or more) of the scalar directions. We will proceed by choosing a cyclic coordinate, $\phi_0$, on which the metric does not depend and Wick rotating along that coordinate. Equivalently, we want to find a coordinate Killing vector, $\partial/\partial\phi_0$. Since $\mathcal{M}$ is a symmetric space there are many Killing vectors to choose from.

\hypertarget{a_general_solution}{}\subsection{{A general solution}}\label{a_general_solution}

The general theory of symmetric spaces\footnote{See, for example, \cite{He}.} tells us how we can find our needed Killing vector. We first choose any 1-parameter subgroup of $G$ or, equivalently, an element of the Lie algebra, $T\in\mathfrak{g}$ and study the isometry of $\mathcal{M}$ given by

\begin{equation}
e^{\phi_0 T}\ .
\label{isometry}\end{equation}
The Minkowski-signature metric is produced by replacing ${d\phi_0}^2\to -{d\phi_0}^2$ in the metric on $\mathcal{M}$.

This procedure necessitates that \eqref{isometry} act on $\mathcal{M}$ without fixed points; otherwise, the Minkowski metric will be singular. A Cartan decomposition of the Lie algebra $\mathfrak{g}$ gives

\begin{equation}
\mathfrak{g} = \mathfrak{k} \oplus \mathfrak{r}
\label{cartanDecomp}\end{equation}
where $\mathfrak{k}$ is the Lie algebra of $K$, a maximal compact subgroup, and

\begin{displaymath}
[\mathfrak{k},\mathfrak{r}]\subset \mathfrak{r},\qquad [\mathfrak{r},\mathfrak{r}]\subset \mathfrak{k}\ .
\end{displaymath}
In addition, the Cartan form is negative-definite on $\mathfrak{k}$
\begin{displaymath}
K(T,T) \lt 0,\quad \forall 0\neq T\in \mathfrak{k}
\end{displaymath}
reflecting the compactness of $K$. The Iwasawa decomposition further decomposes
\begin{equation}
\mathfrak{r} = \mathfrak{a} \oplus \mathfrak{n}
\label{iwasawaDecomp}\end{equation}
where $\mathfrak{a}$ is an abelian subalgebra of $\mathfrak{g}$ and $\mathfrak{n}$ is nilpotent, obeying
\begin{displaymath}
K(T,T) = 0,\quad \forall T\in \mathfrak{n}\ .
\end{displaymath}

We have the following
\begin{prop}
The action on $\mathcal{M}$ defined by $T \in \mathfrak{g}$ is fixed-point free action if $K(T,T) \geq0$.
\end{prop}
\begin{proof} Assume that $g_0K$ is a fixed point of the action, \emph{i.e.},

\begin{displaymath}
e^{s T} g_0 k_1 = g_0 k_2
\end{displaymath}
for $T\neq 0$ and some $k_{1,2}\in K$. It follows that

\begin{displaymath}
g_0^{-1} T g_0 \in \mathfrak{k}\ ,
\end{displaymath}
and hence

\begin{displaymath}
K(g_0^{-1} T g_0,g_0^{-1} T g_0) \lt 0\ .
\end{displaymath}
Since the Killing form is invariant under the adjoint action of $G$, this implies that $K(T,T) \lt 0$. \end{proof}

This is enough to see that the desired Wick rotation exists. We can be more explicit, however. For any real Lie group $G$, we define the generalized transpose

\begin{equation}
T^\# = \begin{cases}-T & T\in\mathfrak{k}\\ T & T\in\mathfrak{r}\end{cases}
\label{genTranspose}\end{equation}
which is equal to minus the Cartan involution of the Lie algebra. It can be extended to the group in a manner that satisfies

\begin{displaymath}
\left(e^X\right)^\# = e^{X^\#}\ .
\end{displaymath}
The invariant metric on $G/K$ is

\begin{equation}
{d s}^2 = \frac{1}{2} Tr (m^{-1} d m)^2
\label{invariantmetric}\end{equation}
where

\begin{displaymath}
m = g g^\#\ .
\end{displaymath}
The transformation

\begin{equation}
g \to e^{\phi_0 T} g
\label{isomGen}\end{equation}
is an isometry of \eqref{invariantmetric}, and the coordinate $\phi_0$ is cyclic. As above, if $K(T,T)\geq 0$, then \eqref{isomGen} is fixed-point free, and we can Wick rotate ${d\phi_0}^2\to -{d\phi_0}^2$ in \eqref{invariantmetric}. If, moreover, $T$ is in the orthogonal complement of $\mathfrak{k}$, then the flow through the identity, $g(\phi_0) = e^{\phi_0 T}$, is a geodesic\footnote{In other words, if we choose $T$ such that $K(T, X) = 0, \forall X\in \mathfrak{k}$, then $e^{s T}$ is a geodesic on $G/K$. This result is originally due to Cartan. See, for example, \cite{Ma}.}  of infinite length in either the Riemannian or Minkowskian signature.

Arkani-Hamed \emph{et al} take $T\in \mathfrak{n}$ which implies that $K(T,T) =0$, but the flow is no longer necessarily geodesic. Moreover, in their case, the timelike geodesics of the Wick-rotated metric are bounded in length, and, in many case of possible interest, fail to satisfy the inequality \eqref{bound}.

\hypertarget{an_example}{}\subsection{{An example}}\label{an_example}

To make things more concrete, let us apply this to the simple example of $\mathcal{M}= SL_{2}(\mathbb{R})/SO(2)$. We can parametrize the coset space as

\begin{displaymath}
g = \exp{\begin{pmatrix} v/2 & 0\\ 0 &-v/2\end{pmatrix}}
\cdot \exp{\begin{pmatrix} 0 & u/2\\ u/2 &0\end{pmatrix}}
\cdot O
\end{displaymath}
where $O$ is an $SO(2)$ matrix. Setting

\begin{displaymath}
m = g g^T = \begin{pmatrix}
e^v \cosh u & \sinh u\\
\sinh u& e^{-v} \cosh u\end{pmatrix}\ ,
\end{displaymath}
the metric on $\mathcal{M}$ is

\begin{equation}
{d s}^2 = \frac{1}{2} Tr (m^{-1} d m)^2 = {d u}^2 + \cosh^2 u {d v}^2\ .
\label{sl2metric}\end{equation}
The coordinate $v$ is cyclic and corresponds to the 1-parameter subgroup generated by $T=\left(\begin{smallmatrix}1&0\\0&-1\end{smallmatrix}\right)\in\mathfrak{a}$. We can Wick rotate $v\to i v$ and obtain a Minkowski-signature metric. Moreover, $u=0, v(\tau)=\tau$ is a geodesic of infinite length in either signature.

We can perform the following change of coordinates

\begin{equation}
\rho = \sqrt{\frac{\cosh u \cosh v -1}{\cosh u \cosh v +1}},\quad \sin\theta = \frac{\cosh u \sinh v}{\sqrt{\cosh^{2} u \cosh^{2} v -1}}\ .
\label{discCoords}\end{equation}
This gives the usual Poincar\'e metric on the unit disk

\begin{equation}
{d s}^2 = \frac{4({d\rho}^2 +\rho^2 {d\theta}^2)}{(1-\rho^2)^2}
\label{discMetric}\end{equation}
in which $\theta$ is a cyclic coordinate corresponding to the the 1-parameter subgroup generated by $T = \left(\begin{smallmatrix}0 & -1\\ 1 & 0\end{smallmatrix}\right)\in\mathfrak{k}$. One cannot Wick rotate $\theta$, however, as the resulting Minkowsi metric is singular at $\rho=0$.

If we make the further change of coordinates

\begin{equation}
\rho e^{i\theta} = \frac{i z +1}{z +i}\ ,
\label{uhpCoords}\end{equation}
we obtain the familiar metric on the upper half plane

\begin{equation}
{d s}^2 = \frac{d z\ d\overline{z}}{{(\mathrm{Im\ } z)}^2}
\label{UHP}\end{equation}
where $x=\mathrm{Re\ }z$ is cyclic corresponding to the nilpotent generator, $T = \left(\begin{smallmatrix}0&1\\ 0 & 0\end{smallmatrix}\right)\in \mathfrak{n}$. We obtain the desired Minkowski-signature metric by Wick-rotating $x\to i x$. Timelike geodesics in the resulting Minkowski-signature metric take the form

\begin{displaymath}
\begin{aligned}
  x(\tau) &= x_0 - y_0 \tan\tau\ , \\
  y(\tau) &= \frac{y_0}{\cos\tau}
\end{aligned}
\end{displaymath}
with $-\pi/2 \lt \tau \lt \pi/2$ and have length $\delta\tau = \pi$.

\hypertarget{the_question_of_uduality}{}\section{{The question of U-duality}}\label{the_question_of_uduality}

We have seen that, in contrast to some claims in the literature (reiterated by Arkani-Hamed \textit{et al}), the Euclidean supergravity theory has complex wormhole solutions for any $d\leq 9$. We must then ask the question: do such complex saddle points of the Euclidean action contribute to the path integral? Arkani-Hamed \emph{et al} adduced evidence from AdS/CFT that they do not, at least for the case of $AdS_3\times S^2\times T^4$. We would like to argue, more generally, that they never contribute.

As discussed above, the true moduli space for the scalar fields in string theory compactified on $T^d$ is

\begin{equation}
\mathcal{M}_0 = G(\mathbb{Z})\setminus G/K
\label{trueM}\end{equation}
where the U-duality group, $G(\mathbb{Z})$, acts a discrete gauge symmetry of the theory.

There is no problem studying geodesics on $\mathcal{M}_0$ by looking at geodesics on the covering space. This is not sufficient, however, as we are interested in performing a Wick rotation on the geodesic. We will present evidence that this will, in general, be incompatible with U-duality invariance.

Note that we do not mean ``incompatible'' in some trivial sense. Any particular scalar field configuration, whether it is $\phi=\text{const}$ or a ``rolling'' configuration, $\phi=\phi(r)$, such as we are considering, is not invariant under the U-duality group. Instead, U-duality maps one such configuration into another. In particular, it maps solutions of the supergravity equations of motion into one another.

This, rather trivial, lack of invariance is not what is at issue. Indeed, if the example of D-instantons is any guide, the resolution is to sum over saddle points. Any particular D-instanton breaks S-duality modular invariance, but the sum has the correct modular properties under S-duality (see, e.g. \cite{Green:1997di}).

Our problem here is that we are interested in finding a lift of the U-duality group to some ``complexification'' (loosely speaking) of M or, better, some complexification of the space of scalar field configurations. As before, any particular \emph{complex} configuration of $\phi$ will not be invariant under U-duality. But one expects that U-duality will map one such configuration into another (and we might hope that a suitable sum over such configurations restores U-duality invariance).

This, we argue, is what fails to be the case. It is that failure that we we are referring to when we claim that complexifying $\mathcal{M}$ is incompatible with U-duality.

In special cases, complexification might be compatible with some subgroup of the U-duality group. In the case of the nilpotent generator, \eqref{UHP}, one finds that the subgroup

\begin{displaymath}
\left\{\begin{pmatrix}1&n\\0& 1\end{pmatrix}, n\in\mathbb{Z}\right\}\subset SL_{2}(\mathbb{Z})
\end{displaymath}
is preserved by the Wick rotation. For the more interesting case of $T\in \mathfrak{a}$, however, the U-duality group is broken completely.

We suspect that there is no extension of the discrete gauge symmetry of string theory to complex values of the fields. This means that we cannot ``deform the contour'' (to use the familiar metaphor of steepest descent) to pick up these saddle points as the semiclassical approximation to the path integral. For this reason, these complex saddle points of the supergravity theory cannot contribute to the quantum gravity path integral.

We can see what the problem is more concretely for the simplest case of $SL_{2}(\mathbb{R})/SO(2)$. Our choices of Wick rotations consists of choosing a basepoint, $g_0 K\in \mathcal{M}$, and an element, $T$, of the Lie algebra (up to scale), satisfying $K(T,T) \geq 0$. Such a choice gives rise to a cyclic coordinate, $\phi_0$. Our ``partial complexification'' of $\mathcal{M}$ then replaces the real variable $\phi_0$ with a complex one, $\tilde{\phi}_0$. We obtain in this way a 3-manifold, $\widetilde{\mathcal{M}}$, with a \emph{nonsingular} complex bilinear form on its tangent space which is our original metric with ${d\phi_0}^2$ replaced by ${d\tilde{\phi}_0}^2$. The bilinear form so obtained reduces to our original metric \eqref{invariantmetric} on the subspace $\tilde{\phi}_0 \in \mathbb{R}$ and to the ``Wick-rotated'' Minkowskian metric for $\tilde{\phi}_0 \in i\mathbb{R}$.

What we seek is an action of $SL_{2}(\mathbb{Z})$ on $\widetilde{\mathcal{M}}$ which

\begin{itemize}%
\item reduces to the usual action of $SL_{2}(\mathbb{Z})$ on $\mathcal{M}\subset\widetilde{\mathcal{M}}$, and
\item is an isometry of the complex bilinear form on $\widetilde{\mathcal{M}}$.

\end{itemize}
One can check that this is impossible by a brute force computation. Consider, for example, the nilpotent case, \eqref{UHP}. Writing $z=x+i y$, the action of $SL_{2}(\mathbb{Z})$ on $\mathcal{M}$ is

\begin{equation}
\begin{split}
  x &\to \frac{(a x+b)(c x+d)+ a c y^2 }{(c x+d)^2 +(c y)^2}\ ,\\
  y &\to \frac{y}{(c x+d)^2 +(c y)^2}
\end{split}
\label{sl2z}\end{equation}
for $\left(\begin{smallmatrix}a&b\\c&d\end{smallmatrix}\right)\in SL_{2}(\mathbb{Z})$. We would like to generalize this to complex $x$ in such a way that the result is an isometry of the complex bilinear form,

\begin{displaymath}
{d s}^2 = \frac{ {d x}^2 + {d y}^2}{y^2}\ .
\end{displaymath}
The latter demands that \eqref{sl2z} depend holomorphically on $x$. This implies, however, that the transformation for $y$ in \eqref{sl2z} is not real unless $c=0$. Thus, only this subgroup of the U-duality group $SL_{2}(\mathbb{Z})$ is compatible with this complexification.

Similar arguments hold in the hyperbolic case \eqref{sl2metric}.

More geometrically, we can think of the above ``partial complexification'' as follows. Choosing a $T$ allows us to write $\mathcal{M}$ as a real line bundle over base $B$ ($\partial/\partial\phi_0$ being a vertical tangent vector). $\widetilde{\mathcal{M}}$ is obtained by complexifying the fibers, resulting in a complex line bundle over the same base.

To generalize, we could Wick rotate on more than one cyclic coordinate. Thus, we give $\mathcal{M}$ the structure of a real \emph{vector} bundle

\begin{displaymath}
\begin{aligned}
V \to&\mathcal{M}\\
&\downarrow\\
&B
\end{aligned}
\end{displaymath}
of rank $k\leq r$, such that the metric on $\mathcal{M}$ takes the form

\begin{displaymath}
{d s}^2_{\mathcal{M}} = {d s}^2_B + h
\end{displaymath}
where $h$ is an metric on the vertical tangent space. $\widetilde{\mathcal{M}}$ is constructed by complexifying the fibers,

\begin{displaymath}
\begin{aligned}
V\otimes \mathbb{C} \to&\widetilde{\mathcal{M}}\\
&\downarrow\\
&B
\end{aligned}
\end{displaymath}
and extending $h$ to a $\mathbb{C}$-bilinear form.

In this context, our objective is to find a lift of the $G(\mathbb{Z})$ action on $\mathcal{M}$ to an action on the partial complexification, $\widetilde{\mathcal{M}}$ such that it acts by isometries of the complex bilinear form on $\widetilde{\mathcal{M}}$. We have not found an elegant proof of this result, but it seems clear that it follows from the above explicit computation by restricting to an $SL_{2}$ subgroup.

In closing, we should mention a possible loophole in this argument. Rather than restricting ourselves to complexifying cyclic coordinates, we could pick some particular coordinate system and complexify everything. This would not make sense for a general real manifold, but because $G/K$ is contractible, it can be covered in a single coordinate chart. For the example of $SL_{2}(\mathbb{R})/SO(2)$, we could choose the upper half plane coordinates, $x,y$  \eqref{UHP}, or the $u,v$ coordinates \eqref{sl2metric}, and promote them to complex variables. If we promote the $SL_{2}(\mathbb{Z})$ symmetry to act holomorphically on these variables, then we seem to have an answer to the above objection. Each such choice of coordinates clearly leads to a (very!) different complexification. But, even more problematically, each one seems to lead to some pathology. In the upper half plane coordinates, $y$ was supposed to be positive, and it is unclear what the proper range of values should be when we complexify it. The $u,v$ coordinates run over all real values, so there is no a-priori problem when complexifying with letting them run over all complex values. But, if we do so, then the metric \eqref{sl2metric} has singularities at $u\in i\pi\mathbb{Z}$. So, even liberalizing the rules of what it might mean to ``complexify'' the field space, $G/K$, does not appear to lead to any satisfactory solution.

\section*{Acknowledgements}

J.D.~would  like to thank J. Polchinski for discussions.  Both authors would like to thank the Aspen Center for Physics, where this work was completed.

\end{document}